\documentclass[10pt, conference]{IEEEtran}
\IEEEoverridecommandlockouts
\usepackage{cite}
\usepackage{amsmath,amssymb,amsfonts}
\usepackage{graphicx}
\usepackage{textcomp}
\usepackage{xcolor}
\usepackage{bbm}
\usepackage{comment}
\usepackage{algorithm2e}
\usepackage{tikz, pgfplots}
\pgfplotsset{compat=1.18}

\usetikzlibrary{positioning}
\usetikzlibrary{arrows.meta}
\usetikzlibrary{shapes.geometric}
\usetikzlibrary{intersections}
\usetikzlibrary{calc}
\def\BibTeX{{\rm B\kern-.05em{\sc i\kern-.025em b}\kern-.08em
    T\kern-.1667em\lower.7ex\hbox{E}\kern-.125emX}}

\usepackage{booktabs}
\usepackage{array}
\usepackage{siunitx}

\DeclareMathOperator*{\argmin}{arg\,min}
\DeclareMathOperator{\sign}{sign}

\newtheorem{prop}{Proposition}
\newtheorem{remark}{Remark}

\makeatletter
\newcommand{\linebreakand}{%
  \end{@IEEEauthorhalign}
  \hfill\mbox{}\par
  \mbox{}\hfill\begin{@IEEEauthorhalign}
}
\makeatother

\title{Efficient Demodulation for Coset Codes on Lattices\\

\thanks{This research was supported in part by the U.S. Department of Commerce’s National Telecommunications and Information Administration (NTIA) under the Public Wireless Supply Chain Innovation Fund Grant Program (Award 24-60-IF2415: ASPEN - Advanced Signal Processing Enhancement for Next-Generation Open Radio Units), administered by the National Institute of Standards and Technology.}
}

\author{
\IEEEauthorblockN{Chloe Makdad}
\IEEEauthorblockA{\textit{Rampart Communications}\\
Linthicum Heights, MD, USA \\
cmakdad@rampartcommunications.com}
\and
\IEEEauthorblockN{Leopold Bertholet}
\IEEEauthorblockA{\textit{Rampart Communications}\\
Linthicum Heights, MD, USA \\
lbertholet@rampartcommunications.com}
\and
\IEEEauthorblockN{Mairead Cummins}
\IEEEauthorblockA{\textit{Rampart Communications}\\
Linthicum Heights, MD, USA \\
mairead@rampartcommunications.com}
\linebreakand
\IEEEauthorblockN{Stephen Mackes}
\IEEEauthorblockA{\textit{Rampart Communications}\\
Linthicum Heights, MD, USA \\
stephen@rampartcommunications.com}
\and
\IEEEauthorblockN{Matthew Robinson}
\IEEEauthorblockA{\textit{Rampart Communications}\\
Linthicum Heights, MD, USA \\
matt@rampartcommunications.com}
}

\begin{document}

% \addtolength{\textfloatsep}{-10pt}
% \addtolength{\abovecaptionskip}{-1pt}

\maketitle

\begin{abstract}
Lattice constellations can provide significant coding and shaping gains over traditional quadrature amplitude modulations (QAMs), but have not been widely adopted due to the high implementation complexity of modulation, demodulation, and decoding. Voronoi constellations (VCs) enable an efficient modulation map that achieves shaping gain using the Voronoi cell of a lattice. Multilevel coset codes (MLCCs) are a recently proposed coded modulation scheme for VCs that realizes more efficient decoding. However, MLCCs still suffer from a soft demodulation complexity of $O(d2^{sd})$ for a $d$-dimensional lattice with spectral efficiency $s$ bits per real dimension. In this paper, we propose fast quantize-walk demodulation (FQWD), an efficient demodulation scheme for MLCCs that reduces the complexity to $O(sd2^d)$. Our strategy utilizes novel radial projections onto the Voronoi cells of the $D_4$ and $E_8$ lattices. Experimental results demonstrate that FQWD has only negligible loss in block error rate (BLER) relative to maximum likelihood demodulation under the additive white Gaussian noise (AWGN) channel. Using FQWD, we simulate MLCCs on higher-dimensional lattices with $s>2$ bits per dimension under AGWN for the first time, achieving BLER gains over multilevel polar codes on QAMs of up to 0.6 dB for the $D_4$ lattice and 1.2 dB for $E_8$.
\end{abstract}

\section{Introduction}

\IEEEPARstart{T}{he} use of higher-dimensional lattices to build signal constellations has long been known to provide gain over traditional quadrature amplitude modulations (QAMs) \cite{forney_efficient_modulation}. There are two ways to use lattices to determine constellations. First, Voronoi constellations (VCs) consist of any constellation shaped by the Voronoi cell of a dense lattice. Second, lattice-valued constellations are subsets of shifted dense lattices which can have any shaping region. Lattice-valued constellations exhibit coding gain over QAMs, especially in the absence of forward error correction (FEC), or when the FEC has a high code rate \cite{TLDC}, \cite{li2025coded}, \cite{frey2020twostage}. Each approach to lattice-based constellations can be applied independently: \cite{li2025coded} presents VCs valued on the integers, while \cite{frey2020twostage} describes lattice-valued constellations that are shaped by hypercubes. In this paper, we follow \cite{conway1983fast} in shaping a lattice-valued constellation with its own Voronoi cell, resulting in both geometric shaping gain and lattice coding gain over QAMs in the additive white Gaussian noise (AWGN) channel.

Despite their performance advantages, lattice constellations have found limited applicability due to their high implementation complexity, especially when used in conjunction with FEC. For high-dimensional lattices, the brute-force algorithms for modulation, demodulation, and decoding are often not practically feasible \cite{forney1988coset}. For lattices with a fast quantizer, the problem of efficient modulation was resolved by the bit-to-symbol mapping proposed in \cite{conway1983fast}. The use of sphere decoding offers a potential approach to demodulation via lattice search \cite{viterbo1999universal}. However, sphere decoding is inherently designed for infinite lattices and is not naturally adapted to Voronoi shaping regions. Efficient decoding can be achieved with multilevel coset codes (MLCCs) \cite{CosetCodes}, a recently proposed coded modulation scheme for VCs after the manner of \cite{forney1988coset}. For a $d$-dimensional constellation with a spectral efficiency of $s$ bits per real dimension and $N$ constellation points per codeword, the decoder operates with complexity $O(s2^d\cdot N\log N)$, improving over a naive decoding complexity of $O(2^{sd}\cdot N\log N)$. However, MLCCs still suffer from computationally expensive demodulation, requiring $O(d2^{sd})$ operations per lattice point.

In this work, we propose fast quantize-walk demodulation (FQWD), a new demodulation strategy for MLCCs. We apply FQWD to the $D_4$ and $E_8$ lattices, achieving a complexity of $O(sd2^d)$ operations per received lattice point. The resulting coded modulation achieves fast modulation, demodulation, and decoding, making it feasible for certain dense lattices. The proposed algorithm relies on a fast quantizer for the chosen lattice, as well as an efficient Voronoi projection function that maps received points in $\mathbb{R}^d$ to the nearest point of the lattice Voronoi cell $\mathcal{V}$. We approximate the Voronoi projection function by novel radial projection functions on each of the $D_4$ and $E_8$ lattices. The functions project a point $\bar{y}$ onto $\mathcal{V}$ along a line segment connecting $\bar{y}$ to the origin. Using the geometry of the $D_4$ and $E_8$ lattices, we give efficient implementations of their radial projection functions. Finally, we present block error rate (BLER) curves for MLCCs comparing FQWD to maximum likelihood (ML) soft demodulation. Results show that FQWD displays negligible loss with appropriate parameter selection.

\section{Background}

We outline the lattice constellation and modulation map used in this work before describing MLCCs. Then, we review existing work on the quantizers needed for FQWD.

\subsection{Fast Lattice Modulation}\label{sec:Boomerang}

The lattice modulation map $\Phi$ was originally introduced in \cite{conway1983fast}, and takes in integers in $\mathbb{Z}/r\mathbb{Z}$ for some $r = 2^s$. In our implementation, bits are first grouped into tuples of size $s$, and mapped to integers in $\mathbb{Z}/r\mathbb{Z}$ via the canonical binary expansion. Then, $\Phi$ maps a tuple of $d$ mod-$r$ integers to a point in a full-rank lattice $\Lambda \subset \mathbb{R}^d$:
\begin{align}
    \Phi: \left(\mathbb{Z}/r\mathbb{Z}\right)^d \to \Lambda + \bar{p} \subset \mathbb{R}^d. \label{eq:modulation}
\end{align}
Note that the lattice $\Lambda$ is shifted by a fixed vector $\bar{p}$, chosen to minimize the average power of the constellation $\mathcal{C}$. To achieve Voronoi shaping, all constellation points lie in $r\mathcal{V}$, where $\mathcal{V}$ is the Voronoi cell of $\Lambda$. Thus, in an MLCC, the constellation points are shifted to a subset of $\Lambda + \bar{p}$, while the shaping region remains centered at the origin.

Before shaping and the shift by $\bar{p}$, the mapping from mod-$r$ integers to lattice points is given by the lattice basis matrix $M$. As a result, ignoring the shift by $\bar{p}$, $\Phi$ descends to a group isomorphism $\tilde{\Phi}: \left( \mathbb{Z}/r\mathbb{Z} \right)^d \to \Lambda/ r\Lambda$, given by $\tilde{\Phi}(\bar{x}) = M\bar{x} + r\Lambda$. Then, for each lattice point, the shift by $\bar{p}$ is applied, and the shaping step selects a coset representative in $r\mathcal{V}$.

\subsection{Multilevel Coset Codes}

MLCCs are a generalization of multilevel codes (MLCs) initially proposed in \cite{CosetCodes}, and are based on \cite{forney1988coset}. They arise from the lattice partition chain
\[
    \Lambda \supset 2\Lambda \supset 4\Lambda \supset \cdots \supset r\Lambda.
\]
Using the modulation map (\ref{eq:modulation}), an MLCC has $s$ levels, where level $\ell$ encodes bits on the lattice quotient $\Lambda^{(\ell)} = 2^{\ell}\Lambda / 2^{\ell + 1}\Lambda$ for $0 \leq \ell < s$. In particular, level $\ell = 0$ encodes the least significant bits (LSBs) of the integer vector inputs to $\Phi$. Since for all $\ell$, $\Lambda^{(\ell)} \cong \Lambda / 2 \Lambda$, each level is treated as a code on $\Lambda/2\Lambda$ with $r=2$, rather than on $\Lambda/r\Lambda$.

Multi-stage decoding (MSD) decodes these bit levels from LSB to most significant bits (MSBs). While the vector of log-likelihood ratios (LLRs) or the probability mass function (PMF) for a lattice symbol in $\Lambda/r\Lambda$ requires $r^d$ entries, each decoding level operates on soft information of size $2^d$ \cite[II.C]{CosetCodes}. The multilevel approach thereby enables a decoding speedup. However, the level-$\ell$ soft information for a constellation point $\bar{z}$ sent under AWGN is obtained by repeatedly marginalizing and conditioning soft information of size $r^d$ over the full lattice. Given a received $\bar{y} \in \mathbb{R}^d$, the full PMF is
\begin{align}
    P(\bar{z} = \bar{x}\mid \bar{y}) \propto \exp\left(-\frac{1}{2\sigma^2}
    \lVert\bar{y} - \bar{x}\rVert^2\right) \label{eq:full_demod}
\end{align}
for each $\bar{x}\in \mathcal{C}$. As a result, the demodulation step remains a bottleneck for MLCCs, with complexity $O(dr^d)$. FQWD addresses this by computing the soft information on $\Lambda^{(\ell)}$ for level $\ell$ directly from received data, bypassing the need for soft information on $\Lambda/r\Lambda$.

 \subsection{Lattice and Coset Quantizers} \label{sec:quantizers}

Both FQWD and the efficient implementation of the modulation map (\ref{eq:modulation}) require a fast quantization function $Q_\Lambda: \mathbb{R}^d \to \Lambda$ for the lattice. Specifically, the quantizer can be used to test if $r\mathcal{V}$ contains a point $\bar{y} \in \mathbb{R}^d$ by checking if $\bar{y}/r$ lies in $\mathcal{V}$, or equivalently, if $Q_\Lambda (\bar{y}/r) = \bar{0}$. For our purposes, a quantizer input $\bar{y}$ will lie on the boundary of $\mathcal{V}$ with probability zero, so the definition of $Q_\Lambda$ on that set can be chosen arbitrarily.

FQWD further makes use of quantizers for shifted sublattices of $\Lambda$. Given $\bar{h} \in \Lambda$ and $n \in \mathbb{Z}$, the quantizer to $\bar{h} + n\Lambda$ is given by \cite{conway1983fast} as
\begin{align}
    Q_{\bar{h}+n\Lambda} (\bar{y}) = n Q_\Lambda \left(\frac{\bar{y} - \bar{h}}{n}\right) + \bar{h}. \label{eq:subset_quantizer}
\end{align}
Given a received point $\bar{y}$ and a fixed $\bar{h}$, an MLCC's level-$\ell$ demodulator must compute the distance from $\bar{y}$ to the closest representative $R_{\bar{x}}$ of $\bar{h} + \bar{x} + 2^{\ell + 1}\Lambda$ for each coset ${\bar{x} \in \Lambda^{(\ell)}}$. The log-sum-exp approximation is used to calculate the LLR, which is why only the closest representative is considered. When the shaping region is ignored, the LLR for each $\bar{x}$ is proportional to the squared distance $\| \bar{y} - R_{\bar{x}} \|^2$, up to a sign. However, this incurs significant loss. In our algorithm, we still select the closest member of each coset, but additionally enforce that the quantized point is in $r\mathcal{V}$.

\section{Fast Demodulation on Lattices}

% The points of the lsb = (0, 1) coset that we want to make xs 

\begin{figure}
    \centering
    \scalebox{.95}{

    \begin{tikzpicture}[point/.style = {fill = black, 
                                        minimum size = 0.15cm,
                                        shape = circle,
                                        inner sep = 0pt,
                                        node font = \tiny
                        },
                        every label/.style={font=\footnotesize},
                        >=Latex]

        \def\cosetatwocoords{1.0/0.0, 0.0/1.732, -1.0/0.0, 0.0/-1.732, 3.0/0.0, 2.0/1.732}
                        
        % Shaping reigon
        \node[shape = regular polygon, regular polygon sides = 6, draw, fill = blue!10, minimum size = 5cm] (hex) at (0, 0) {};
        \node[text = blue, below = of hex.corner 2] {$r\mathcal{V}$};
        % Side closest to received point
        \path [name path = hexSide] (hex.corner 1) -- (hex.corner 6);

        % this is the coset where lsbs = (1, 0)
        \foreach \x/\y [count=\z] in \cosetatwocoords{
            \node[font=\tiny] (lat\z) at (\x, \y) {$\mathbf{\times}$};
        } 
        
        % 0 vector, revieved point, and the ray between them
        \node[point, label=below:$\bar{0}$] (zero) at (0, 0) {};
        \node[point, label=below:$\bar{y}$] (rec_y) at (4, 2) {};
        \draw[name path = steppingLine, dotted, gray] (zero) -- (rec_y);

        % Projected point
        \path [name intersections={of=steppingLine and hexSide, by=intZ}];
        \node[point, label = below:$\bar{u}_0$] (proj_z) at (intZ) {};

        % Stepping points
        \path (proj_z.center) -- (zero.center)
            node[point, pos = 0.2, label = below:$\bar{u}_1$] (z1) {}
            node[point, pos = 0.4, label = above:$\bar{u}_2$] (z2) {};

        % Quantized points
        \node[point, fill = none, draw = red, label={[text=red]above:$\bar{q}_0{,}$ $\bar{q}_1$}] (q1) at (lat6) {};
        \node[point, fill = none, draw = green!55!black, label={[text=green!55!black]below:$\bar{q}_2$}] (q2) at (lat1) {};
        \draw[->, red, dashed] (proj_z) -- (q1);
        \draw[->, red, dashed] (z1) -- (q1);
        \draw[->, green!55!black, dashed] (z2) -- (q2);

    \end{tikzpicture}
    }
    \caption{Illustration of an example of FQWD on the $A_2$ lattice. The received point $\bar{y}$ is projected onto the boundary at $\bar{u}_0 = \pi_{r\mathcal{V}}(\bar{y})$. The algorithm steps along the radial ray with a step size of $\bar{w}$ to points $\bar{u}_k = \bar{u}_0 - k\bar{w}$. Each stepped point is quantized to the $A_2$ lattice points within the desired coset ($\times$) until a representative $\bar{q}_k$ lands inside the shaping region $r\mathcal{V}$.}
    \label{fig:fast_demod_geometry}
\end{figure}
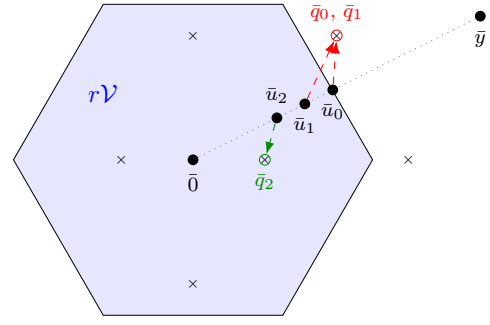

We introduce the radial Voronoi projection functions before describing the FQWD algorithm, treating the specific functions as a black box. Then, we derive the necessary radial projection functions for the $D_4$ and $E_8$ lattices.

\subsection{Voronoi Projection Functions}\label{sec:voronoi_projection}

The FQWD algorithm makes use of a projection function from $\mathbb{R}^d$ onto the boundary of a Voronoi cell. For a full-rank lattice $\Lambda \subset \mathbb{R}^d$, the Voronoi cell $\mathcal{V}$ is always a closed convex polyhedral solid with finitely many faces. These properties guarantee that the nearest-point projection function
\begin{align*}
    \tilde{\pi}_{\mathcal{V}}(\bar{y}) = \argmin_{\bar{v} \in \mathcal{V}}{\lVert \bar{y} - \bar{v}\rVert}
\end{align*} is computable. To achieve low complexity, we instead present fast radial projections that approximate $\tilde{\pi}_{\mathcal{V}}(\bar{y})$.

The radial projection $\pi_{\mathcal{V}} : \mathbb{R}^d \to \mathcal{V}$ is defined by $\pi_{\mathcal{V}}(\bar{y}) = \lambda \bar{y}$ where $\lambda=1$ if $\bar{y} \in \mathcal{V}$. Otherwise, if $\bar{y} \notin \mathcal{V}$, then $\lambda \bar{y}$ is the intersection of the boundary of $\mathcal{V}$ and the line segment between the origin and $\bar{y}$. We provide an efficient implementation of this function for the $D_4$ and $E_8$ lattices in Sections \ref{sec:D4voronoi} and \ref{sec:E8voronoi}. 

\subsection{System Model}

Let $M$ be the generator matrix for $\Lambda$ and let $\bar{p}$ be the shift vector in the modulation map (\ref{eq:modulation}). In an MLCC codeword, each sent lattice point $\bar{z} \in \Lambda + \bar{p}$ is treated as a sum of $s$ components $\bar{z}^{(i)} \in \Lambda^{(i)}$:
\begin{align*}
    \bar{z} = \bar{z}^{(0)} + \bar{z}^{(1)} + \cdots + \bar{z}^{(s -1)} + \bar{p}.
\end{align*}
During stage $\ell$ of decoding, the decoder takes the corresponding noisy lattice point $\bar{y}\in\mathbb{R}^d$ received under AWGN with noise variance $\sigma^2$, and estimates component $\bar{z}^{(\ell)}$ given known values for the previous components $\bar{h} = \bar{z}^{(0)} + \cdots + \bar{z}^{(\ell-1)} \in \Lambda/2^{\ell}\Lambda$, where we set $\bar{h} = \bar{0}$ for level $\ell = 0$. The soft information input to decoding stage $\ell$ is the probability of each possible value $\bar{x} \in \Lambda^{(\ell)}$ for $\bar{z}^{(\ell)}$:
\begin{align}
    P(\bar{z}^{(\ell)} = \bar{x} \mid \bar{y}, \bar{h}) \propto \exp\left(-\frac{1}{2\sigma^2}D[\bar{x}]\right), \label{eq:pmf}
\end{align}
where $D[\bar{x}] = \lVert\bar{y} - \bar{q}_{\bar{x}}\rVert^2$, and $\bar{q}_{\bar{x}}$ is the representative of $\bar{x} + \bar{h} + \bar{p} + 2^{\ell+1}\Lambda$ closest to $\bar{y}$ in the constellation. This LLR is output for each $\bar{x} \in \Lambda^{(\ell)}$, and is used by the stage-$\ell$ decoder .

\subsection{The FQWD Algorithm}

The proposed algorithm is shown in Algorithm \ref{alg:demod}, and an example on the $A_2$ lattice is illustrated in Figure \ref{fig:fast_demod_geometry}.

\SetFuncSty{textrm}
\RestyleAlgo{ruled}
\SetCommentSty{itshape}
\SetArgSty{textnormal}
\begin{algorithm}[t]
\caption{FQWD for received point $\bar{y}$, level $\ell$, and known previous-level values $\bar{h}$.}\label{alg:demod}
\SetKwRepeat{DoWhile}{do}{while}
\SetKwProg{Fn}{Function}{\{}{\}}
\SetKwFunction{FDemod}{Demodulate}
\SetKwFunction{FCosetQuantize}{CosetQuantize}
\SetKwFunction{FShapingCheck}{ShapingCheck}
\SetAlgoLined % Turn off lines for outer function block
\Fn{\FDemod{$\bar{y}$, $\ell$, $\bar{h}$}}{
    \SetAlgoVlined % Turn lines back on
    \tcp{Project received to shaping region}
    $\bar{v} \gets \pi_{r\mathcal{V}}(\bar{y})$\;
    \tcp{Define step size}
    $\epsilon \gets 2^{\ell+1} c_\Lambda/K$\;
    $\bar{w} \gets \epsilon \bar{v} / \|\bar{v}\|$\;
    \tcp{Store squared distance to each coset}
    $D[\bar{x}] \gets 0$, $\bar{x} \in \Lambda^{(\ell)}$\;
    \tcp{Track which cosets have representatives identified}
    $C[\bar{x}] \gets \text{false}$, $\bar{x} \in \Lambda^{(\ell)}$\;

    $\bar{u} \gets\bar{v}$\;
    \While{not all $C[\bar{x}]$ true}{
        \ForEach{$\bar{x} \in \Lambda^{(\ell)}$ such that $!C[\bar{x}]$}{
            \tcp{Find nearest coset representative}
            $\bar{q}_{\bar{x}} \gets R_{\bar{x}}^{(\ell)}(\bar{u}\mid\bar{h})$\;
            \tcp{Check if lies in shaping region}
            \If{$\bar{q}_{\bar{x}} \in r\mathcal{V}$}{
                $C[\bar{x}] \gets \text{true}$\;
                $D[\bar{x}] = \|\bar{y} - \bar{q}_{\bar{x}}\|^2$\;
            }
        }
        \tcp{Avoid stepping past origin}
        \eIf{$\lVert \bar{w} \rVert < \lVert \bar{u} \rVert$}{
            \tcp{Push $\bar{u}$ closer to origin}
            $\bar{u} \gets \bar{u} - \bar{w}$\;
        }
        {
            $\bar{u} \gets \bar{0}$\;
        }
    }

    \Return{$D$}\;
    % \SetAlgoNoLine % Turn off lines again
    \SetAlgoLined
}
\end{algorithm}

The FQWD algorithm first projects $\bar{y}$ into $r\mathcal{V}$, then iteratively quantizes and takes steps towards the origin if the quantized point is not within $r\mathcal{V}$. Let $\epsilon = \frac{2^{\ell+1} c_\Lambda}{K}$, where $c_\Lambda$ denotes the covering radius of $\Lambda$, defined in \cite{conway1999sphere}. The maximum number of steps is $K$, so $K+1$ is the maximum number of iterations. For good performance, $K=64$ is seen to be sufficient, though some configurations permit smaller $K$ values, as described later. Given a shift from previous levels $\bar{h}$, let $R_{\bar{x}}^{(\ell)}(\bar{u}\mid\bar{h})$ denote the closest point to $\bar{y}$ within a level-$\ell$ coset $\bar{x} \in \Lambda^{(\ell)}$. It is computed using the quantizer for shifted sublattices (\ref{eq:subset_quantizer}):
\[
    R_{\bar{x}}^{(\ell)}(\bar{u}\mid\bar{h}) = Q_{\bar{x}+\bar{h} + \bar{p} + 2^{\ell+1}\Lambda}(\bar{u}).
\]

The algorithm proceeds as follows. Let $\bar{v} = \pi_{r\mathcal{V}}(\bar{y}) = r\pi_{\mathcal{V}}(\bar{y}/r)$ be the initial radial projection of the received point to the shaping region, as defined in Section \ref{sec:voronoi_projection}. We set $\bar{w} = \epsilon \bar{v}/\|\bar{v}\|$, which will be used to repeatedly step into the shaping region. Initialize the evaluation point $\bar{u} \gets \bar{v}$.

For each coset $\bar{x} \in \Lambda^{(\ell)}$ compute the closest candidate $\bar{q}_{\bar{x}} = R_{\bar{x}}^{(\ell)}(\bar{u}\mid\bar{h})$ in that lattice, given the known shifts $\bar{p}$ and $\bar{h}$. Then, verify whether $\bar{q}_{\bar{x}}$ is in $r\mathcal{V}$ using the lattice quantizer $Q_{\Lambda}$, as described in Section \ref{sec:quantizers}. If $\bar{q}_{\bar{x}} \in r\mathcal{V}$, set $\bar{q}_{\bar{x}}$ as the optimal representative, and mark coset $\bar{x}$ completed. While any cosets remain incomplete, update $\bar{u} \gets \bar{u} - \bar{w}$ step-by-step along the radial ray between $\bar{0}$ and $\bar{y}$ until $\bar{q}_{\bar{x}} \in r\mathcal{V}$. If $\bar{u}$ steps past the origin, we set it to $\bar{0}$, where the algorithm is guaranteed to terminate. We can detect when this case occurs by checking if $\lVert \bar{w} \rVert > \lVert \bar{u} \rVert$, or more simply, by checking if any nonzero coordinate of $\bar{u}$ changes sign. Finally, we set $D[\bar{x}] = \| \bar{y} - \bar{q}_{\bar{x}}\|^2$.

\begin{prop}[Algorithm Termination]
    After $K$ steps, each coset $\bar{x}$ will have a representative $\bar{q}_{\bar{x}} \in r\mathcal{V}$.
\end{prop}
\begin{IEEEproof}
Assume first that $\bar{u}$ does not step past the origin. Then, at the $K^\text{th}$ step, the value of $\bar{u}$ is
\begin{align*}
    \bar{u}_K &= \bar{v} - K\bar{w} = \bar{v}\left(1 - \frac{2^{\ell+1}c_\Lambda}{\lVert \bar{v} \rVert}\right).
\end{align*}
Since $\bar{v} \in r\mathcal{V}$ and $\lVert \bar{v} \rVert \leq rc_{\Lambda}$, we can write $\bar{u}_K = \alpha\bar{v}$, where
\begin{align}
    \alpha& \leq  1 - \frac{2^{\ell+1}c_\Lambda}{r c_\Lambda} =1- \frac{1}{2^{s-\ell-1}}. \label{eq:alpha_ineq}
\end{align}
Because $\ell < s$, inequality (\ref{eq:alpha_ineq}) also applies when $\bar{u}_K$ is set to $\bar{0}$, in which case $\alpha=0$.

Let $\mathcal{S} = r\mathcal{V}$ denote the shaping region. The nearest representative $\bar{q}_{\bar{x}}$ to $\bar{u}_K$ must lie in the scaled Voronoi cell centered at $\bar{u}_K$,
\begin{align*}
    \bar{q}_{\bar{x}} \in \bar{u}_K + 2^{\ell + 1}\mathcal{V} = \alpha \bar{v} + \beta \mathcal{S}
\end{align*}
where $\beta = 1/2^{s-\ell-1}$. Since $\mathcal{S}$ is a convex set containing the origin, $\bar{v} \in \mathcal{S}$, and $\alpha,\beta$ are nonnegative scalars satisfying $\alpha + \beta \leq 1$,
\begin{align*}
    \alpha \bar{v} + \beta \mathcal{S} \subseteq \mathcal{S}.
\end{align*}
Hence, $\bar{q}_{\bar{x}}$ lies in $\mathcal{S}$.
\end{IEEEproof}

\begin{remark}[Lossless MSB speedup] \label{rem:fqwd_speedup} Note that at the final level $\ell = s - 1$, each coset contains only a single representative within $r\mathcal{V}$. For implementation efficiency, one can bypass steps $0, \dots, K-1$ by evaluating directly at $\bar{u}_K = \bar{0}$, or execute ML demodulation at this top level.
\end{remark}

\begin{remark}[Values for $K$] While FQWD must produce a representative in $r\mathcal{V}$ for each coset after $K$ steps, it is not guaranteed to be the nearest one. If a selected representative is not the nearest of its coset to $\bar{y}$, the calculated LLRs will be incorrect, resulting in a less accurate decode. In practice, we demonstrate negligible loss on $D_4$ in BLER compared to ML when $K=16$ for $r=4$, or $K=64$ otherwise.
\end{remark}

\subsection{The Projection Function for $D_4$}\label{sec:D4voronoi}
We now construct the radial projection for the Voronoi cell of the $D_4$ lattice as generated by the columns of
\begin{align*}
    M = \left( \begin{array}{c c c c}
        1 & 1 & 0 & 0 \\
        1 & -1 & 1 & 0 \\
        0 & 0 & -1 & 1 \\
        0 & 0 & 0 & -1
    \end{array} \right).
\end{align*}
The $D_4$ lattice has 24 short vectors, consisting of all vectors in $\mathbb{R}^4$ with exactly two nonzero entries that are each $\pm1$ \cite{conway1999sphere}. Let $\bar{y} = (y_1 ,y_2, y_3, y_4) \in \mathbb{R}^4$ be any point other than the origin and $y_{n},y_{m}$ be the two largest coordinates in absolute value, with ties broken arbitrarily. Let $\bar{f} \in \mathbb{R}^4$ be the vector that is $0$ everywhere except in positions $n$ and $m$, where it has values $\sign (y_{n})$ and $\sign (y_{m})$, respectively. Then, $\bar{f}$ is one of the short vectors of the $D_4$ lattice that is closest to $\bar{y}$, because the dot product of any short vector with $\bar{y}$ is maximized by $\bar{f}$.

Since $D_4$ is a root lattice, its short vectors are in bijection with the faces of its Voronoi cell $\mathcal{V}$ \cite{conway1999sphere}. For any short vector $\bar{b}$, $\frac{1}{2}\bar{b}$ is the centroid of the corresponding Voronoi face. This means that the ray from the origin through $\bar{y}$ intersects the face $\mathcal{F}$ corresponding to $\bar{f}$ at $\lambda_0\bar{y}$ for some $\lambda_0 > 0$. If $\bar{f}$ was the unique closest short vector, then the intersection is in the interior of $\mathcal{F}$, otherwise it is on a boundary of $\mathcal{F}$. Using that $\operatorname{proj}_{\langle \bar{f} \rangle}(\lambda_0\bar{y}) = \bar{f}/2$, we obtain
\begin{align}
    \lambda_0 = \dfrac{1}{2} \dfrac{\bar{f}\cdot\bar{f}}{\bar{f}\cdot\bar{y}} = \dfrac{1}{|y_n| + |y_m|}. \label{eq:lambda0_d4}
\end{align}
Note that if $\lambda_0 \ge 1$, $\bar{y} \in \mathcal{V}$. Thus, we can set $\lambda = \min \left( \lambda_0, 1\right)$, and guarantee $\lambda \bar{y}$ is the closest point in $\langle \bar{y} \rangle \cap \mathcal{V}$ to $\bar{y}$, resulting in $\pi_{\mathcal{V}} (\bar{y}) = \lambda \bar{y}$ as the radial projection function to $\mathcal{V}$.

\subsection{The Projection Function for $E_8$}\label{sec:E8voronoi}

The radial projection to the Voronoi cell $\mathcal{V}$ of the $E_8$ lattice is computed similarly to $D_4$. Again, we use that because $E_8$ is a root lattice, its short vectors are in bijection with the faces of $\mathcal{V}$. We use the basis for $E_8$ given by the columns of
{\small
\begin{align*}
    M=\left( \begin{array}{c c c c c c c c}
        1 & 1 & 0 & 0 & 0 & 0 & 0 & 1/2 \\
        1 & -1 & 1 & 0 & 0 & 0 & 0 & 1/2 \\
        0 & 0 & -1 & 1 & 0 & 0 & 0 & 1/2 \\
        0 & 0 & 0 & -1 & 1 & 0 & 0 & 1/2 \\
        0 & 0 & 0 & 0 & -1 & 1 & 0 & 1/2 \\
        0 & 0 & 0 & 0 & 0 & -1 & 1 & 1/2 \\
        0 & 0 & 0 & 0 & 0 & 0 & -1 & 1/2 \\
        0 & 0 & 0 & 0 & 0 & 0 & 0 & 1/2
    \end{array} \right).
\end{align*}
}

The short vectors of $E_8$ are all either short vectors of $D_8$, implying that they are $0$ everywhere except in two coordinates where they are $\pm 1$, or they are $\pm \frac{1}{2}$ everywhere with an even number of negative signs \cite{conway1999sphere}. To identify the short vector $\bar{f}$ closest to a received $\bar{y}$, we generate the optimal choices $\bar{f}_0$ from the $D_8$ short vectors and $\bar{f}_1$ from the remaining vectors before selecting between the two.

Let $\bar{y} \in \mathbb{R}^8$ be a point other than the origin. Let $y_{n}$ and $y_{m}$ be the two entries of $\bar{y}$ with largest absolute value. Then $\bar{f}_0$ is the short vector with all $0$ entries except in positions $n$ and $m$, where it has values $\sign(y_n)$ and $\sign(y_m)$ respectively. 

To find $\bar{f_1}$, first define $\bar{s}$ by $s_i = \sign(y_i) \in \{\pm 1\}$, with $\sign(0)$ defined arbitrarily. If $\bar{s}$ contains an odd number of $-1$ entries, identify the index $j = \argmin_i |y_i|$ corresponding to the coordinate with smallest magnitude $c = |y_j|$ and flip the sign of that entry: $s_j \mapsto -s_j$. This guarantees an even number of negative signs. If $\bar{s}$ already contained an even number of $-1$ entries, set $c=0$. Then, define $\bar{f}_1 = \frac{1}{2} \bar{s}$. 

Finally, choose $\bar{f}$ to be whichever of $\bar{f}_0$ or $\bar{f}_1$ maximizes the dot product with $\bar{y}$, using the simplifications that 
\begin{align*}
\bar{f}_0 \cdot \bar{y} = |y_n|+|y_m| \quad\text{and}\quad \bar{f}_1 \cdot \bar{y} = \frac{1}{2} \sum_{i=1}^8 |y_i| - c.
\end{align*}

After constructing $\bar{f}$ efficiently, we can compute
\begin{align*}
    \lambda_0 = \dfrac{1}{2}\dfrac{\bar{f} \cdot \bar{f}}{\bar{f} \cdot \bar{y}} = \dfrac{1}{\max\{\bar{f}_0 \cdot \bar{y}, \bar{f}_1 \cdot \bar{y} \}}
\end{align*}
as before. Then setting $\lambda = \min \left( \lambda_0,1\right)$, we have $\pi_{\mathcal{V}}(\bar{y}) = \lambda \bar{y}$. We observe that for both $D_4$ and $E_8$, the radial projection function involves only simple arithmetic operations.

\section{Results and Analysis}

We now evaluate FQWD by comparing its computational complexity against ML soft demodulation, followed by BLER performance evaluations in the AWGN channel. In particular, using FQWD provides a substantial complexity improvement with negligible loss in BLER, enabling previously intractable simulations with high $r$ values and $E_8$ lattice constellations.

\subsection{Complexity of FQWD}

\begin{table*}[ht!]
\centering
\caption{Exact Operation Count per Symbol for ML and FQWD}
\label{tab:complexity_comparison}
\footnotesize
\setlength{\tabcolsep}{3pt}
\begin{tabular}{ccc rrr rrr rrr}
\toprule
& & & \multicolumn{3}{c}{\textbf{Real Additions}} & \multicolumn{3}{c}{\textbf{Real Multiplications}} & \multicolumn{3}{c}{\textbf{Quantizations}} \\
\cmidrule(lr){4-6} \cmidrule(lr){7-9} \cmidrule(lr){10-12}
\textbf{Lattice} & $r$ & $K$ & \textbf{ML} & \textbf{FQWD} & \textbf{Ratio $\big(\frac{\text{ML}}{\text{FQWD}}\big)$} & \textbf{ML} & \textbf{FQWD} & \textbf{Ratio} & \textbf{ML} & \textbf{FQWD} & \textbf{Ratio} \\
\midrule
$D_4$ & 4  & 16 & 15,444          & 6,241   & $2.47$             & 17,408          & 132 & $132$             & 256           & 1,024   & $0.250$ \\
      & 8  & 64 & 246,344         & 37,009  & $6.66$             & 278,528         & 196 & $1.42 \times 10^3$ & 4,096         & 6,144   & $0.667 $ \\
      & 16 & 64 & 3,936,528       & 49,345  & $79.8$             & 4,456,448       & 260 & $1.71 \times 10^4$ & 65,536        & 8,192   & $8.00 $ \\
\midrule
$E_8$ & 4  & 64 & 3,954,004       & 393,313 & $10.1$             & 4,456,448       & 132 & $3.38 \times 10^4$ & 65,536        & 65,536  & $1.00 $ \\
      & 8  & 64 & 1,009,029,704   & 589,969 & $1.71 \times 10^3$ & 1,140,850,688   & 196 & $5.82 \times 10^6$ & 16,777,216    & 98,304  & $171 $ \\
      & 16 & 64 & 257,984,368,912 & 786,625 & $3.28 \times 10^5$ & 292,057,776,128 & 260 & $1.12 \times 10^9$ & 4,294,967,296 & 131,072 & $3.28 \times 10^4 $ \\
\bottomrule
\end{tabular}
\end{table*}

The asymptotic computational complexity of ML demodulation is exponential in the lattice dimension $d$ and polynomial in $r$ of degree $d$, scaling as $O(d r^d)$ for total arithmetic operations and quantizations. Assuming efficient $d$-dimensional radial Voronoi projection functions, the complexity of FQWD combined over the $\log_2 r$ levels of MLCC decoding is $O(d K 2^d \log r)$ operations. As $K$ must be chosen large enough to appropriately mitigate loss, the optimal choice of $K$ may depend on $d$. We leave this analysis to future work and report viable $K$ values for the configurations we tested.

We provide counts of the number of real additions, multiplications, and quantizations needed by FQWD and ML in Table \ref{tab:complexity_comparison}. We record data for both the $D_4$ and $E_8$ lattices with $r$ values $4, 8$, and $16$. We use $K=16$ or $64$ depending on $r$ for FQWD and report the number of operations required if all $K$ steps were taken without incorporating the FQWD simplification in Remark \ref{rem:fqwd_speedup}. Note that for each lattice, the quantization process consists of a small number of additions and multiplications.

The ML PMF demodulator computes (\ref{eq:full_demod}) directly, requiring a squared norm for each of the $r^d$ points in the constellation. Those constellation points are either stored in a large lookup table or computed via the modulation described in Section \ref{sec:Boomerang}, which is what we record in Table \ref{tab:complexity_comparison}. Then, it conditions and marginalizes the resulting soft information to obtain soft information for each bit level (\ref{eq:pmf}), requiring addition.

\subsection{BLER Performance under AWGN}

We evaluate the BLER performance of MLCCs in the AWGN channel with three objectives: demonstrating that in an MLCC scheme with MSD, FQWD has negligible loss relative to ML demodulation; illustrating the impact of $K$; and extending the work in \cite{CosetCodes} by showing the performance of MLCCs for $r>4$ on $D_4$ and $E_8$. All decoders are CRC-aided successive cancellation list decoders with list size $8$, and they use the CRC dictated by the 3GPP 5G standard for each given rate $R$ and bit block size $n_{\text{bit}}$ \cite{multiplexing20185g}. We opt to compare with multilevel polar codes and exclude bit-interleaved polar codes, as MLC is known to outperform bit-interleaved coded modulations under AWGN \cite{BICM}. We restrict our MLCCs to the $D_4$ and $E_8$ lattices, and use an $r^2$-QAM for the MLCs to ensure a fair comparison in terms of spectral efficiency. Code design is as described in \cite{CosetCodes}.

Figure \ref{fig:mlvfqwd} shows that, using an appropriate $K$, FQWD demonstrates negligible loss relative to ML demodulation in an MLCC, with a significant reduction in complexity as recorded in Table \ref{tab:complexity_comparison}. Though Figure \ref{fig:mlvfqwd} shows one configuration, we observe this behavior at a wide range of $R$ and $n_{\text{bit}}$ values.

\begin{figure}[!t]
    \centering
    \begin{tikzpicture}
\begin{axis}[
    scale = 0.7,
    name = demod,
    xshift=2cm,
    xlabel={\scriptsize{$E_b/N_0$ (dB)}},
    ylabel={\scriptsize{BLER}},
    xmin=0, xmax=10,
    ymin=1e-4, ymax=1.5,
    ymode=log,
    mark repeat={2.5},
    grid=both,
    grid style={line width=0.2pt, draw=gray!30},
    major grid style={line width=0.4pt, draw=gray!60},
    width=10cm,
    height=8cm,
    mark size=1.5pt,
    % legend pos= outer north east,
    legend style={font=\scriptsize,anchor=center, at={(0.5,-0.4)},legend columns=2}
]

\addplot[
    color=blue,
    mark=*,
    mark options = solid,
    loosely dashed,
    thick,
] coordinates {
    (0, {1})
    (0.25, {1})
    (0.5, {1})
    (0.75, {1})
    (1, {1})
    (1.25, {0.9994292237442922})
    (1.5, {0.9920091324200914})
    (1.75, {0.9803082191780822})
    (2, {0.9189497716894978})
    (2.25, {0.778824200913242})
    (2.5, {0.598458904109589})
    (2.75, {0.3818493150684932})
    (3, {0.209189497716895})
    (3.25, {0.09189497716894977})
    (3.5, {0.03351844506148354})
    (3.75, {0.006551698170042925})
    (4, {0.002011510846759446})
    (4.25, {0.000348132372214941})
    (4.5, {7.167431192660551e-5})
    (4.75, {1.02391874180865e-5})
    
};
\addlegendentry{FQWD, $r =4$}

\addplot[
    color=red,
    mark=*,
    mark options = solid,
    dotted,
    thick,
] coordinates {
    (0, {1})
    (0.25, {1})
    (0.5, {1})
    (0.75, {1})
    (1, {0.9997146118721462})
    (1.25, {0.9997146118721462})
    (1.5, {0.9911529680365296})
    (1.75, {0.9654680365296804})
    (2, {0.9052511415525114})
    (2.25, {0.7699771689497716})
    (2.5, {0.58162100456621})
    (2.75, {0.3661529680365297})
    (3, {0.199486301369863})
    (3.25, {0.08761415525114155})
    (3.5, {0.030859049207673062})
    (3.75, {0.007932436165594537})
    (4, {0.001966289646865008})
    (4.25, {0.000348132372214941})
    (4.5, {8.1913499344692e-5})
    (4.75, {1.535878112712975e-5})
    
};
\addlegendentry{ML, $r = 4$}

\addplot[
    color=blue,
    mark=square*,
    mark options = solid,
    loosely dashed,
    thick,
] coordinates {
    (0, {1})
    (0.5, {1})
    (1, {1})
    (1.5, {1})
    (2, {1})
    (2.5, {1})
    (3, {0.9977266268826371})
    (3.5, {0.9789712986643933})
    (4, {0.868428530832623})
    (4.5, {0.5723216822961068})
    (5, {0.2179596476271668})
    (5.5, {0.038244667810737924})
    (6, {0.004355856973506465})
    (6.5, {0.0003302243605742832})
    (7, {3.4558363315913355e-5})
    (7.5, {3.8398181462125956e-6})

};
\addlegendentry{FQWD, $r = 8$}

\addplot[
    color=red,
    mark=square*,
    mark options = solid,
    dotted,
    thick,
] coordinates {
    (0, {1})
    (0.5, {1})
    (1, {1})
    (1.5, {1})
    (2, {1})
    (2.5, {0.9997155049786629})
    (3, {0.9988620199146515})
    (3.5, {0.9823613086770981})
    (4, {0.8605974395448079})
    (4.5, {0.5852062588904694})
    (5, {0.21621621621621623})
    (5.5, {0.04011461318051576})
    (6, {0.003924401107499699})
    (6.5, {0.0003071807336243871})
    (7, {1.9198795851524193e-5})
    
};
\addlegendentry{ML, $r=8$}

\addplot[
    color=blue,
    mark=triangle*,
    mark options = solid,
    loosely dashed,
    thick,
] coordinates {
    (0, {1})
    (0.5, {1})
    (1, {1})
    (1.5, {1})
    (2, {1})
    (2.5, {1})
    (3, {1})
    (3.5, {1})
    (4, {1})
    (4.5, {1})
    (5, {1})
    (5.5, {0.9954324864401941})
    (6, {0.9349129317727662})
    (6.5, {0.6974022266628604})
    (7, {0.306023408506994})
    (7.5, {0.06666666666666667})
    (8, {0.006501890620676934})
    (8.5, {0.0003484725287489836})
    (9, {5.11959370904325e-6})

};
\addlegendentry{FQWD, $r = 16$}

\addplot[
    color=red,
    mark=triangle*,
    mark options = solid,
    dotted,
    thick,
] coordinates {
   (0, {1})
    (0.5, {1})
    (1, {1})
    (1.5, {1})
    (2, {1})
    (2.5, {1})
    (3, {1})
    (3.5, {1})
    (4, {1})
    (4.5, {1})
    (5, {0.9997145304025121})
    (5.5, {0.9934341992577791})
    (6, {0.931201827005424})
    (6.5, {0.6640022837567799})
    (7, {0.2797602055381102})
    (7.5, {0.05901374292643492})
    (8, {0.004422526248806872})
    (8.5, {0.00017918302778360834})
    (9, {5.1195150795316665e-6})
};
\addlegendentry{ML, $r=16$ }

\end{axis}

\end{tikzpicture}

\vspace{-2em}
    \caption{BLER curves comparing the performance of $D_4$ MLCCs with FQWD and ML demodulations for bit block size $n_{\text{bit}} = 1024$ for $r = 4$ and $16$ and $n_{\text{bit}} = 768$ for $r=8$. All codes have rate $R=1/2$. Using $K=16$ for $r=4$ and $K=64$ for $r=8$ and $16$ shows no noticeable loss from FQWD.}
    \label{fig:mlvfqwd}
\end{figure}
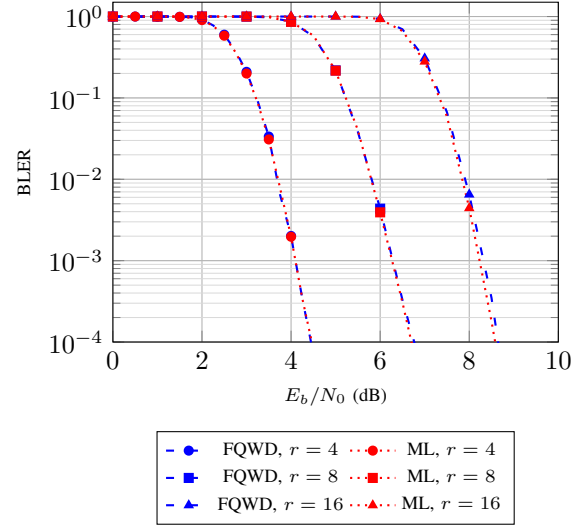

It is worth noting that if $K$ is too small, using FQWD will result in loss compared to ML demodulation. Figure \ref{fig:K} shows that an $r=16$ MLCC with $n_{\text{bit}} = 1024$ and rate $R = 3/4$ using FQWD with $K=4$ yields significant loss when compared to ML in terms of BLER. Increasing $K$ to $16$ reduces the loss significantly, and for $K = 64$ the performance of FQWD is almost indistinguishable from ML.

\begin{figure}
    \centering
    \begin{tikzpicture}
\begin{axis}[
    scale = 0.7,
    name = r8,
    xshift=2cm,
    xlabel={\scriptsize{$E_b/N_0$ (dB)}},
    ylabel={\scriptsize{BLER}},
    xmin=10, xmax=16,
    ymin=1e-4, ymax=1.5,
    ymode=log,
    mark repeat={1},
    grid=both,
    grid style={line width=0.2pt, draw=gray!30},
    major grid style={line width=0.4pt, draw=gray!60},
    width=10cm,
    height=8cm,
    mark size=1.5pt,
    legend pos= north east,
    legend style={font=\scriptsize,legend columns=1}
]

\addplot[
    color=blue,
    mark=*,
    mark options = solid,
    loosely dashed,
    thick,
] coordinates {
   (0, {1})
    (0, {1})
    (0.5, {1})
    (1, {1})
    (1.5, {1})
    (2, {1})
    (2.5, {1})
    (3, {1})
    (3.5, {1})
    (4, {1})
    (4.5, {1})
    (5, {1})
    (5.5, {1})
    (6, {1})
    (6.5, {1})
    (7, {1})
    (7.5, {1})
    (8, {1})
    (8.5, {1})
    (9, {1})
    (9.5, {1})
    (10, {0.9997146118721462})
    (10.5, {0.9877283105022832})
    (11, {0.916095890410959})
    (11.5, {0.6649543378995434})
    (12, {0.3541666666666667})
    (12.5, {0.13955479452054795})
    (13, {0.040872687103010216})
    (13.5, {0.009661364301746852})
    (14, {0.0018780096308186195})
    (14.5, {0.0002611046261596117})
    (15, {6.1436382625791e-5})
    (15.5, {1.535909565644775e-5})
    
};
\addlegendentry{FQWD, $K = 4$}

\addplot[
    color=olive,
    mark=triangle*,
    mark options = solid,
    dashdotted,
    thick,
] coordinates {
    (0, {1})
    (0.5, {1})
    (1, {1})
    (1.5, {1})
    (2, {1})
    (2.5, {1})
    (3, {1})
    (3.5, {1})
    (4, {1})
    (4.5, {1})
    (5, {1})
    (5.5, {1})
    (6, {1})
    (6.5, {1})
    (7, {1})
    (7.5, {1})
    (8, {1})
    (8.5, {1})
    (9, {1})
    (9.5, {1})
    (10, {0.9988584474885844})
    (10.5, {0.9803082191780822})
    (11, {0.8276255707762558})
    (11.5, {0.4531963470319635})
    (12, {0.12014840182648402})
    (12.5, {0.020509019026439337})
    (13, {0.0010941820107065122})
    (13.5, {4.6077286969343246e-5})
    (14, {7.679547828223874e-6})
    
};
\addlegendentry{FQWD, $K = 16$}

\addplot[
    color=magenta,
    mark=square*,
    mark options = solid,
    dotted,
    thick,
] coordinates {
    (0, {1})
    (0.5, {1})
    (1, {1})
    (1.5, {1})
    (2, {1})
    (2.5, {1})
    (3, {1})
    (3.5, {1})
    (4, {1})
    (4.5, {1})
    (5, {1})
    (5.5, {1})
    (6, {1})
    (6.5, {1})
    (7, {1})
    (7.5, {1})
    (8, {1})
    (8.5, {1})
    (9, {1})
    (9.5, {1})
    (10, {0.9982871824150727})
    (10.5, {0.9797316585783614})
    (11, {0.8212960319725949})
    (11.5, {0.4416214673137311})
    (12, {0.09734513274336283})
    (12.5, {0.009660842754367934})
    (13, {0.00038651763388103144})
    (13.5, {1.5357326596586068e-5})

};
\addlegendentry{FQWD, $K = 64$}

\addplot[
    color=cyan,
    mark=star,
    mark options = solid,
    solid,
    thick,
] coordinates {
   (0, {1})
    (0.5, {1})
    (1, {1})
    (1.5, {1})
    (2, {1})
    (2.5, {1})
    (3, {1})
    (3.5, {1})
    (4, {1})
    (4.5, {1})
    (5, {1})
    (5.5, {1})
    (6, {1})
    (6.5, {1})
    (7, {1})
    (7.5, {1})
    (8, {1})
    (8.5, {1})
    (9, {1})
    (9.5, {1})
    (10, {0.997724039829303})
    (10.5, {0.9769559032716928})
    (11, {0.797724039829303})
    (11.5, {0.40056899004267427})
    (12, {0.07795163584637269})
    (12.5, {0.007495096665732698})
    (13, {0.0003678293271921828})
    (13.5, {7.678958118962419e-6})
};
\addlegendentry{ML}

\end{axis}

\end{tikzpicture}

\vspace{-3em}
    \caption{BLER performance of a $D_4$ $r=16$ MLCC with $n_{\text{bit}} = 1024$, $R = 3/4$ demonstrating the performance of ML demodulation compared to FQWD with varying $K$. }
    \label{fig:K}
\end{figure}
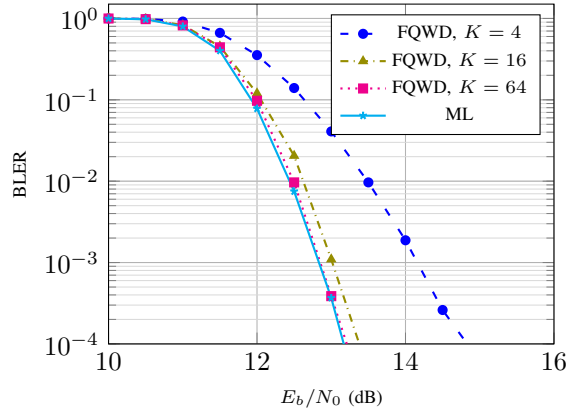

With the reduction in complexity from using FQWD, it is now tractable to simulate MLCCs with $r=8$ and $r=16$, on both $D_4$ and $E_8$. Figures \ref{fig:r8} and \ref{fig:r16} compare those MLCCs to traditional MLC polar codes. Much like the results of \cite{CosetCodes}, we see gains over MLC polar codes across a variety of rates. For $D_4$ MLCCs with both $r=8$ and $r=16$, we observe gains between $0.45 \text{ dB}$ and $0.6 \text{ dB}$. Figure \ref{fig:r16} also shows the first published instance of $E_8$ MLCC results, which achieve gains between $0.8\text{ dB}$ and $1.2\text{ dB}$ over MLC polar codes with $r=16$. This demonstrates shaping gain over $D_4$ at all rates shown, as well as additional coding gain at $R = 15/16$.

\begin{figure}[!t]
    \centering
    \begin{tikzpicture}
\begin{axis}[
    scale = 0.7,
    name = r8,
    xshift=2cm,
    xlabel={\scriptsize{$E_b/N_0$ (dB)}},
    ylabel={\scriptsize{BLER}},
    xmin=3, xmax=13,
    ymin=1e-4, ymax=1.5,
    ymode=log,
    mark repeat={2.5},
    grid=both,
    grid style={line width=0.2pt, draw=gray!30},
    major grid style={line width=0.4pt, draw=gray!60},
    width=10cm,
    height=8cm,
    mark size=1.5pt,
    % legend pos= outer north east,
    legend style={font=\scriptsize,anchor=center, at={(0.5,-0.4)},legend columns=2}
]

\addplot[
    color=blue,
    mark=*,
    mark options = solid,
    loosely dashed,
    thick,
] coordinates {
   (0, {1})
    (0.5, {1})
    (1, {1})
    (1.5, {1})
    (2, {1})
    (2.5, {1})
    (3, {0.9977266268826371})
    (3.5, {0.9789712986643933})
    (4, {0.868428530832623})
    (4.5, {0.5723216822961068})
    (5, {0.2179596476271668})
    (5.5, {0.038244667810737924})
    (6, {0.004355856973506465})
    (6.5, {0.0003302243605742832})
    (7, {3.4558363315913355e-5})
    (7.5, {3.8398181462125956e-6})
    
};
\addlegendentry{$D_4$ MLCC, $R = 1/2$}

\addplot[
    color=magenta,
    mark=*,
    mark options = solid,
    dotted,
    thick,
] coordinates {
    (0, {1})
    (0.5, {1})
    (1, {1})
    (1.5, {1})
    (2, {1})
    (2.5, {1})
    (3, {1})
    (3.5, {0.9997158283603297})
    (4, {0.9917590224495595})
    (4.5, {0.9096334185848253})
    (5, {0.6146632566069906})
    (5.5, {0.23643080420574028})
    (6, {0.0407042786823173})
    (6.5, {0.003853847964226877})
    (7, {0.0001689519984333542})
    (7.5, {1.1519454438637787e-5})
    
};
\addlegendentry{MLC, $R = 1/2$}

\addplot[
    color=blue,
    mark=square*,
    mark options = solid,
    loosely dashed,
    thick,
] coordinates {
    (0, {1})
    (0.5, {1})
    (1, {1})
    (1.5, {1})
    (2, {1})
    (2.5, {1})
    (3, {1})
    (3.5, {1})
    (4, {1})
    (4.5, {1})
    (5, {1})
    (5.5, {1})
    (6, {1})
    (6.5, {0.9982949701619779})
    (7, {0.9701619778346121})
    (7.5, {0.7985223074737141})
    (8, {0.4282466609832339})
    (8.5, {0.10514350667803353})
    (9, {0.015077774596812673})
    (9.5, {0.001194708102759556})
    (10, {9.215156628865326e-5})
    (10.5, {5.759472893040829e-6})

};
\addlegendentry{$D_4$ MLCC, $R = 3/4$}

\addplot[
    color=magenta,
    mark=square*,
    mark options = solid,
    dotted,
    thick,
] coordinates {
    (0, {1})
    (0.5, {1})
    (1, {1})
    (1.5, {1})
    (2, {1})
    (2.5, {1})
    (3, {1})
    (3.5, {1})
    (4, {1})
    (4.5, {1})
    (5, {1})
    (5.5, {1})
    (6, {1})
    (6.5, {1})
    (7, {0.9991474850809889})
    (7.5, {0.975561238988349})
    (8, {0.789712986643933})
    (8.5, {0.388178459789713})
    (9, {0.08019133370849747})
    (9.5, {0.008008757778289929})
    (10, {0.000604744653769287})
    (10.5, {2.8797364465204143e-5})
    
};
\addlegendentry{MLC, $R = 3/4$}

\addplot[
    color=blue,
    mark=triangle*,
    mark options = solid,
    loosely dashed,
    thick,
] coordinates {
   (0, {1})
    (0.5, {1})
    (1, {1})
    (1.5, {1})
    (2, {1})
    (2.5, {1})
    (3, {1})
    (3.5, {1})
    (4, {1})
    (4.5, {1})
    (5, {1})
    (5.5, {1})
    (6, {1})
    (6.5, {1})
    (7, {1})
    (7.5, {1})
    (8, {1})
    (8.5, {0.9960215970446149})
    (9, {0.9462915601023018})
    (9.5, {0.7075873827791986})
    (10, {0.3108837737993748})
    (10.5, {0.06194941744813868})
    (11, {0.0039392234102419805})
    (11.5, {0.000215039311874202})
    (12, {6.719978496068813e-6})

};
\addlegendentry{$D_4$ MLCC, $R = 7/8$}

\addplot[
    color=magenta,
    mark=triangle*,
    mark options = solid,
    dotted,
    thick,
] coordinates {
   (0, {1})
    (0.5, {1})
    (1, {1})
    (1.5, {1})
    (2, {1})
    (2.5, {1})
    (3, {1})
    (3.5, {1})
    (4, {1})
    (4.5, {1})
    (5, {1})
    (5.5, {1})
    (6, {1})
    (6.5, {1})
    (7, {1})
    (7.5, {1})
    (8, {1})
    (8.5, {1})
    (9, {0.9988633134413185})
    (9.5, {0.9698778061949418})
    (10, {0.7774936061381074})
    (10.5, {0.35493037794828075})
    (11, {0.07247579529737207})
    (11.5, {0.00584726932522512})
    (12, {0.0002620791613466837})
    (12.5, {1.3439956992137626e-5})
};
\addlegendentry{MLC, $R = 7/8$ }

\end{axis}

\end{tikzpicture}

\vspace{-3em}
    \caption{BLER results comparing $D_4$ MLCCs with FQWD for $r=8$ to MLC polar codes on a $64$-QAM with $n_{\text{bit}} = 768$.}
    \label{fig:r8}
\end{figure}
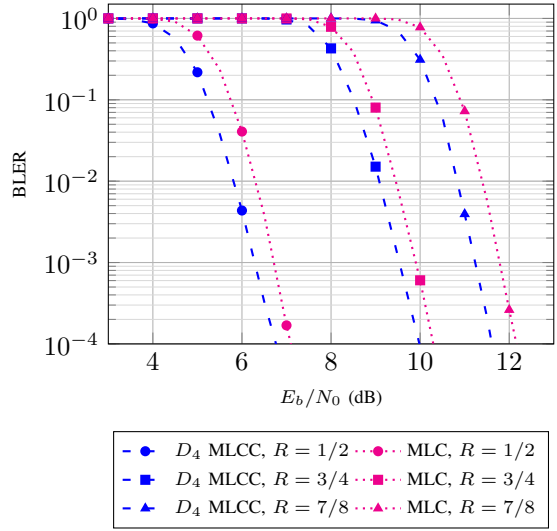

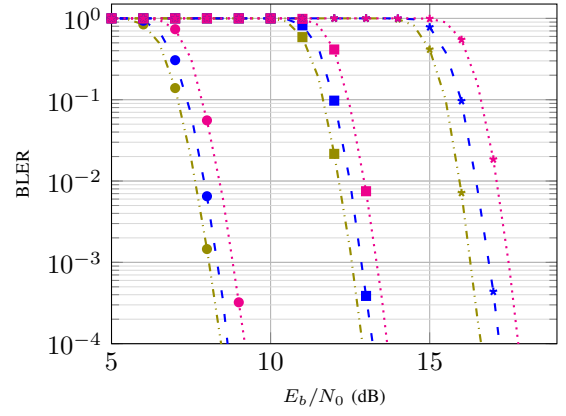
\begin{figure}[!t]
    \centering
    \begin{tikzpicture}
\begin{axis}[
    scale = 0.7,
    name = r16,
    xshift=2cm,
    xlabel={\scriptsize{$E_b/N_0$ (dB)}},
    ylabel={\scriptsize{BLER}},
    xmin=5, xmax=19,
    ymin=1e-4, ymax=1.5,
    ymode=log,
    mark repeat={2.5},
    grid=both,
    grid style={line width=0.2pt, draw=gray!30},
    major grid style={line width=0.4pt, draw=gray!60},
    width=10cm,
    height=8cm,
    mark size=1.5pt,
    % legend pos= outer north east,
    legend style={font=\scriptsize,anchor=center, at={(0.4,-0.4)},legend columns=3}
]

\addplot[
    color=olive,
    mark=*,
    mark options = solid,
    dashdotdotted,
    thick,
] coordinates {
    (0, {1})
    (0.5, {1})
    (1, {1})
    (1.5, {1})
    (2, {1})
    (2.5, {1})
    (3, {1})
    (3.5, {1})
    (4, {1})
    (4.5, {0.9997158283603297})
    (5, {0.9991474850809889})
    (5.5, {0.9752770673486786})
    (6, {0.8422847399829497})
    (6.5, {0.49701619778346123})
    (7, {0.13895993179880647})
    (7.5, {0.01916682964991199})
    (8, {0.001457011707669297})
    (8.5, {7.167871551741793e-5})
    (9, {5.119908251244138e-6})
    
};
\addlegendentry{$E_8$ MLCC, $R = \frac{1}{2}$}

\addplot[
    color=blue,
    mark=*,
    mark options = solid,
    loosely dashed,
    thick,
] coordinates {
    (0, {1})
    (0.5, {1})
    (1, {1})
    (1.5, {1})
    (2, {1})
    (2.5, {1})
    (3, {1})
    (3.5, {1})
    (4, {1})
    (4.5, {1})
    (5, {1})
    (5.5, {0.9954324864401941})
    (6, {0.9349129317727662})
    (6.5, {0.6974022266628604})
    (7, {0.306023408506994})
    (7.5, {0.06666666666666667})
    (8, {0.006501890620676934})
    (8.5, {0.0003484725287489836})
    (9, {5.11959370904325e-6})
    
};
\addlegendentry{$D_4$ MLCC, $R = \frac{1}{2}$}

\addplot[
    color=magenta,
    mark=*,
    mark options = solid,
    dotted,
    thick,
] coordinates {
    (0, {1})
    (0.5, {1})
    (1, {1})
    (1.5, {1})
    (2, {1})
    (2.5, {1})
    (3, {1})
    (3.5, {1})
    (4, {1})
    (4.5, {1})
    (5, {1})
    (5.5, {1})
    (6, {0.9982871824150727})
    (6.5, {0.9608906651441621})
    (7, {0.7356551527262346})
    (7.5, {0.3200114187838995})
    (8, {0.05566644488800177})
    (8.5, {0.006153595370540172})
    (9, {0.000322529450010495})
    (9.5, {1.5358545238595e-5})
    
};
\addlegendentry{MLC, $R = \frac{1}{2}$}

\addplot[
    color=olive,
    mark=square*,
    mark options = solid,
    dashdotdotted,
    thick,
] coordinates {
    (0, {1})
    (0.5, {1})
    (1, {1})
    (1.5, {1})
    (2, {1})
    (2.5, {1})
    (3, {1})
    (3.5, {1})
    (4, {1})
    (4.5, {1})
    (5, {1})
    (5.5, {1})
    (6, {1})
    (6.5, {1})
    (7, {1})
    (7.5, {1})
    (8, {1})
    (8.5, {1})
    (9, {1})
    (9.5, {0.9997145304025121})
    (10, {0.995717956037682})
    (10.5, {0.9172138167285184})
    (11, {0.5880673708250072})
    (11.5, {0.1881244647445047})
    (12, {0.02157332752233602})
    (12.5, {0.0014276002719238613})
    (13, {4.60719797897582e-5})
    (13.5, {0})

};
\addlegendentry{$E_8$ MLCC, $R = \frac{3}{4}$}

\addplot[
    color=blue,
    mark=square*,
    mark options = solid,
    loosely dashed,
    thick,
] coordinates {
    (0, {1})
    (0.5, {1})
    (1, {1})
    (1.5, {1})
    (2, {1})
    (2.5, {1})
    (3, {1})
    (3.5, {1})
    (4, {1})
    (4.5, {1})
    (5, {1})
    (5.5, {1})
    (6, {1})
    (6.5, {1})
    (7, {1})
    (7.5, {1})
    (8, {1})
    (8.5, {1})
    (9, {1})
    (9.5, {1})
    (10, {0.9982871824150727})
    (10.5, {0.9797316585783614})
    (11, {0.8212960319725949})
    (11.5, {0.4416214673137311})
    (12, {0.09734513274336283})
    (12.5, {0.009660842754367934})
    (13, {0.00038651763388103144})
    (13.5, {1.5357326596586068e-5})

};
\addlegendentry{$D_4$ MLCC, $R = \frac{3}{4}$}

\addplot[
    color=magenta,
    mark=square*,
    mark options = solid,
    dotted,
    thick,
] coordinates {
    (0, {1})
    (0.5, {1})
    (1, {1})
    (1.5, {1})
    (2, {1})
    (2.5, {1})
    (3, {1})
    (3.5, {1})
    (4, {1})
    (4.5, {1})
    (5, {1})
    (5.5, {1})
    (6, {1})
    (6.5, {1})
    (7, {1})
    (7.5, {1})
    (8, {1})
    (8.5, {1})
    (9, {1})
    (9.5, {1})
    (10, {1})
    (10.5, {1})
    (11, {0.9857265201256066})
    (11.5, {0.8307165286896946})
    (12, {0.4142163859548958})
    (12.5, {0.08094435075885328})
    (13, {0.007490499853843906})
    (13.5, {0.0003532185117214795})
    (14, {7.678663298293034e-6})
    
};
\addlegendentry{MLC, $R = \frac{3}{4}$}

\addplot[
    color=olive,
    mark=star,
    mark options = solid,
    dashdotdotted,
    thick,
] coordinates {
   (0, {1})
    (0.5, {1})
    (1, {1})
    (1.5, {1})
    (2, {1})
    (2.5, {1})
    (3, {1})
    (3.5, {1})
    (4, {1})
    (4.5, {1})
    (5, {1})
    (5.5, {1})
    (6, {1})
    (6.5, {1})
    (7, {1})
    (7.5, {1})
    (8, {1})
    (8.5, {1})
    (9, {1})
    (9.5, {1})
    (10, {1})
    (10.5, {1})
    (11, {1})
    (11.5, {1})
    (12, {1})
    (12.5, {1})
    (13, {1})
    (13.5, {0.9985726520125606})
    (14, {0.9794461889808735})
    (14.5, {0.8175849272052527})
    (15, {0.41450185555238367})
    (15.5, {0.10048529831572937})
    (16, {0.007154129405576012})
    (16.5, {0.00020468524526409514})
    (17, {9.59775796373967e-6})

};
\addlegendentry{$E_8$ MLCC, $R = \frac{15}{16}$}

\addplot[
    color=blue,
    mark=star,
    mark options = solid,
    loosely dashed,
    thick,
] coordinates {
    (0, {1})
    (0.5, {1})
    (1, {1})
    (1.5, {1})
    (2, {1})
    (2.5, {1})
    (3, {1})
    (3.5, {1})
    (4, {1})
    (4.5, {1})
    (5, {1})
    (5.5, {1})
    (6, {1})
    (6.5, {1})
    (7, {1})
    (7.5, {1})
    (8, {1})
    (8.5, {1})
    (9, {1})
    (9.5, {1})
    (10, {1})
    (10.5, {1})
    (11, {1})
    (11.5, {1})
    (12, {1})
    (12.5, {1})
    (13, {1})
    (13.5, {1})
    (14, {0.9988581216100485})
    (14.5, {0.9743077362260919})
    (15, {0.7796174707393663})
    (15.5, {0.38452754781615756})
    (16, {0.09648872395089923})
    (16.5, {0.008017402113113735})
    (17, {0.0004368126267969986})
    (17.5, {9.59775796373967e-6})

};
\addlegendentry{$D_4$ MLCC, $R =\frac{15}{16}$}

\addplot[
    color=magenta,
    mark=star,
    mark options = solid,
    dotted,
    thick,
] coordinates {
    (0, {1})
    (0.5, {1})
    (1, {1})
    (1.5, {1})
    (2, {1})
    (2.5, {1})
    (3, {1})
    (3.5, {1})
    (4, {1})
    (4.5, {1})
    (5, {1})
    (5.5, {1})
    (6, {1})
    (6.5, {1})
    (7, {1})
    (7.5, {1})
    (8, {1})
    (8.5, {1})
    (9, {1})
    (9.5, {1})
    (10, {1})
    (10.5, {1})
    (11, {1})
    (11.5, {1})
    (12, {1})
    (12.5, {1})
    (13, {1})
    (13.5, {1})
    (14, {1})
    (14.5, {1})
    (15, {0.9891521552954611})
    (15.5, {0.8935198401370255})
    (16, {0.5463888095917785})
    (16.5, {0.1592920353982301})
    (17, {0.018497438816163916})
    (17.5, {0.000883120872369835})
    (18, {1.9198279834126863e-5})

};
\addlegendentry{MLC, $R = \frac{15}{16}$ }

\end{axis}

\end{tikzpicture}

\vspace{-3em}
    \caption{BLER results comparing $r = 16$ $E_8$ and $D_4$ MLCCs to MLC polar codes on a $256$-QAM with $n_{\text{bit}}=1024$.}
    \label{fig:r16}
\end{figure}

\section{Discussion and Future Work}

We have proposed FQWD, a fast quantizer-based demodulator for multilevel coset codes that reduces the demodulation complexity from $O(dr^d)$ to $O(dK  2^d\log r)$ per lattice symbol. The complexity in $d$ assumes an $O(d)$ radial projection operation, which is immediate for all $D_n$ lattices following the work in \ref{sec:D4voronoi}. FQWD resolves the demodulation bottleneck of MLCCs, aligning the asymptotic complexity order with that of the decoder. We developed efficient radial Voronoi projection functions for $D_4$ and $E_8$ and demonstrated that incorporating the bounded shaping constraint $r\mathcal{V}$ into demodulation can be achieved without brute-force enumeration.

While FQWD has negligible loss under the AWGN channel, its current form uses the lattice quantizer, and hence the geometry of the Voronoi cell. In fading channels, channel phase and amplitude variations distort the Voronoi boundary, impacting direct radial projection toward the origin and use of the coset quantizer. Future work could extend the current framework to fading channels by incorporating adjustments to the nearest coset representatives based on the channel state. An additional area for future investigation is determining whether $K$ must grow with $d$ to maintain performance.

\bibliographystyle{IEEEtran}
\bibliography{bibliography}

@unpublished{CosetCodes,
  author    = {Bertholet, Leopold and Makdad, Chloe and Mackes, Stephen and Chew, Daniel and Robinson, Matthew},
  title     = {Multilevel Coset Codes on Lattices},
  note      = {To appear in \textit{Proceedings of the 60th Asilomar Conference on Signals, Systems, and Computers}},
  year      = {2026},
  url       = {https://arxiv.org/abs/2604.06125}
}

@article{li2025coded,
  title={Coded modulation schemes for Voronoi constellations},
  author={Li, Shen and Mirani, Ali and Karlsson, Magnus and Agrell, Erik},
  journal={IEEE Transactions on Communications},
  year={2025},
  publisher={IEEE}
}

@article{forney1988coset,
  title={Coset codes. I. Introduction and geometrical classification},
  author={Forney, G David},
  journal={IEEE Transactions on Information Theory},
  volume={34},
  number={5},
  pages={1123--1151},
  year={1988},
  publisher={IEEE}
}

@article{frey2020twostage,
  author    = {Frey, Felix and Stern, Sebastian and Fischer, Johannes K. and Fischer, Robert F. H.},
  journal   = {Journal of Lightwave Technology}, 
  title     = {Two-Stage Coded Modulation for Hurwitz Constellations in Fiber-Optical Communications}, 
  year      = {2020},
  volume    = {38},
  number    = {11},
  pages     = {2863-2873},
  doi       = {10.1109/JLT.2020.2974241}
}

@unpublished{TLDC,
  author    = {Bertholet, Leopold and Makdad, Chloe and Mackes, Stephen and Chew, Daniel and Robinson, Matthew},
  title     = {Two-Level Decorrelated Coded Modulation on the {$D_4$} Lattice}, 
  note      = {To appear in \textit{Proceedings of the 60th Asilomar Conference on Signals, Systems, and Computers}},
  year      = {2026},
  url       = {https://arxiv.org/abs/2608.16662}
}

@book{conway1999sphere,
  author    = {Conway, John H. and Sloane, Neil J. A.},
  title     = {Sphere Packings, Lattices and Groups},
  edition   = {3rd},
  publisher = {Springer},
  year      = {1999},
  series    = {Grundlehren der mathematischen Wissenschaften},
  volume    = {290},
  isbn      = {978-0-387-98585-5},
  doi       = {10.1007/978-1-4757-6568-7}
}

@ARTICLE{conway1983fast,
  author={Conway, J. and Sloane, N.},
  journal={IEEE Transactions on Information Theory}, 
  title={A fast encoding method for lattice codes and quantizers}, 
  year={1983},
  volume={29},
  number={6},
  pages={820-824},
  doi={10.1109/TIT.1983.1056761}}

@ARTICLE{BICM,
  author={Caire, G. and Taricco, G. and Biglieri, E.},
  journal={IEEE Transactions on Information Theory}, 
  title={Bit-interleaved coded modulation}, 
  year={1998},
  volume={44},
  number={3},
  pages={927-946},
  doi={10.1109/18.669123}}

@techreport{multiplexing20185g,
    author = "{3GPP}",
    title = "{NR; Multiplexing and channel coding}",
    institution = "3rd Generation Partnership Project (3GPP)",
    series = "TS",
    number = "38.212",
    version = "16.5.0"
}

@article{viterbo1999universal,
  author={Viterbo, Emanuele and Boutros, Joseph},
  journal={IEEE Transactions on Information Theory},
  title={A universal lattice code decoder for fading channels}, 
  year={1999},
  volume={45},
  number={5},
  pages={1639-1642},
  doi={10.1109/18.771238}
}

@ARTICLE{forney_efficient_modulation,
  author={Forney, G. and Gallager, R. and Lang, G. and Longstaff, F. and Qureshi, S.},
  journal={IEEE Journal on Selected Areas in Communications}, 
  title={Efficient Modulation for Band-Limited Channels}, 
  year={1984},
  volume={2},
  number={5},
  pages={632-647},
  doi={10.1109/JSAC.1984.1146101}}

\end{document}